\documentclass[letterpaper, 10 pt, conference]{ieeeconf}  
\IEEEoverridecommandlockouts                              
\usepackage[dvipsnames]{xcolor}
\usepackage{arydshln}
\usepackage{graphicx} 
\usepackage{amsmath}
\usepackage{amssymb}
\usepackage{mathtools}
\mathtoolsset{showonlyrefs=true}  
\usepackage{xcolor}
\usepackage{epstopdf}

\newtheorem{theorem}{Theorem}
\newtheorem{corollary}{Corollary}
\newtheorem{proposition}{Proposition}
\newtheorem{lemma}{Lemma}
\newtheorem{definition}{Definition}
\newtheorem{assumption}{Assumption}
\newtheorem{remark}{Remark}

\allowdisplaybreaks

\newcommand{\sh}[1]{\textcolor{RedViolet}{[#1]\raise 0.5ex \hbox{\footnotesize{SH}}}}

\title{\LARGE \bf
Bridging Finite and Infinite-Horizon Nash Equilibria in Linear Quadratic Games
}
\author{Giulio Salizzoni, Sophie Hall, and Maryam Kamgarpour
\thanks{Giulio Salizzoni and Maryam Kamgarpour are with SYCAMORE Lab, School of Engineering, EPFL, Lausanne, Switzerland (email: {\tt\small giulio.salizzoni@epfl.ch}, {\tt\small maryam.kamgarpour@epfl.ch}). Giulio Salizzoni is gratefully supported by the Swiss National Science Foundation, grant number 207984. Sophie Hall is with the Automatic Control Laboratory (IfA) at ETH Zürich, Switzerland (email: {\tt\small shall@control.ee.ethz.ch}) }
}

\begin{document}
\maketitle
\thispagestyle{empty}
\pagestyle{empty}

\begin{abstract}
Finite-horizon linear quadratic (LQ) games admit a unique Nash equilibrium, while infinite-horizon settings may have multiple. We clarify the relationship between these two cases by interpreting the finite-horizon equilibrium as a nonlinear dynamical system. Within this framework, we prove that its fixed points are exactly the infinite-horizon equilibria and that any such equilibrium can be recovered by an appropriate choice of terminal costs. We further show that periodic orbits of the dynamical system, when they arise, correspond to periodic Nash equilibria, and we provide numerical evidence of convergence to such cycles. Finally, simulations reveal three asymptotic regimes: convergence to stationary equilibria, convergence to periodic equilibria, and bounded non-convergent trajectories. These findings offer new insights and tools for tuning finite-horizon LQ games using infinite-horizon analysis.
\end{abstract}

\section{Introduction}

The linear quadratic regulator (LQR) problem, defined by linear system dynamics and quadratic cost functions, is one of the most widely studied and applied models in control theory. In both finite and infinite-horizon formulations, the problem admits a unique optimal solution under standard assumptions, and the two formulations are tightly connected. Specifically, the backward Riccati recursion of the finite-horizon problem converges to the unique solution of the algebraic Riccati equation, which characterizes the infinite-horizon optimal policy. Thus, finite-horizon problems provide a consistent approximation of the infinite-horizon one, and the terminal cost only affects the transient part of the solution

Linear quadratic (LQ) games extend the LQR problem to a multi-agent setting, where each agent interacts with a shared linear system but optimizes an individual cost function defined by its own matrices $Q^i$ and $R^i$. Applications include robotics \cite{cappello2021distributed, losey2019robots}, autonomous driving \cite{fisac2019hierarchical, ren2025chance}, and electricity markets \cite{paccagnan2016aggregative, kirsch2025distributed}. In this work we focus on discrete-time general-sum LQ games and linear state-feedback Nash equilibria. Unlike the LQR, however, the analysis of existence, uniqueness, and stability of Nash equilibria in LQ games is substantially more challenging.

For finite-horizon discrete-time LQ games, it is known that under standard assumptions, a unique feedback Nash equilibrium exists \cite{Basar98}. The infinite-horizon case, in contrast, is far richer. Both \cite{salizzoni2024nash, nortmann2023feedback} demonstrated that the scalar two-agent setting can admit between one and three distinct Nash equilibria. Additionally, \cite{nortmann2024feedback} showed that the number of equilibria can grow exponentially with the number of agents in the scalar case. Extending these results to more general cases poses significant challenges. Even in simpler settings, such as potential decoupled LQ games, characterizing the Nash equilibria remains a highly complex task \cite{PotentialDecoupledLQ}.

The difference in the number of Nash equilibria between the finite-horizon (where a unique Nash equilibrium exists) and infinite-horizon cases (where multiple equilibria may arise) has been observed in the literature \cite{Basar98, engwerda1998open}, but not well understood. It is well-established that one of the infinite-horizon equilibria can be obtained as the limit of the finite-horizon equilibrium as the horizon tends to infinity, if the limit exists \cite{Basar98}. A key result by \cite{weeren99} further demonstrated that, in the continuous-time scalar two-agent setting, this limit always exists and its value depends on the choice of terminal costs in the finite-horizon game. However, they analyzed a specific class of systems and could not generalize the results to higher-dimensional or multi-agent cases.

This discrepancy raises fundamental questions:
\begin{enumerate}
    \item Which infinite-horizon Nash equilibria can be recovered from the finite-horizon games?
    \item What are the possible asymptotic behaviors of the Riccati recursion?
    \item When does the finite-horizon recursion converge?
\end{enumerate}
These questions are both theoretically significant, as they expose structural differences between single and multi-agent settings, and practically relevant, since infinite-horizon formulations are often approximated with finite-horizon strategies.

Our results connects with receding-horizon games \cite{hall2022receding, hall2024receding, hall2024stability, benenati2024linear}. In such settings, agents repeatedly solve finite-horizon LQ games with shifting horizons to approximate infinite-horizon objectives. Whether the resulting strategies converge, and to which equilibria, is critical for ensuring stability and predictability of the closed-loop system. Our results also provide insight into incentive design: while in classical control the cost weights are tuned to shape closed-loop behavior, in games the matrices $Q^i$ and $R^i$ typically represent agents’ preferences. Nonetheless, they can still be influenced in hierarchical or incentive-aligned designs \cite{coogan2013energy, kessler2020linear}, where understanding the mapping between finite and infinite-horizon equilibria can be useful.

The main contributions of the paper are threefold:
\begin{enumerate}
    \item By viewing the Riccati recursion of finite-horizon LQ games as a discrete-time dynamical system, we prove that any infinite-horizon Nash equilibrium can be obtained as the outcome of a finite-horizon game through the choice of appropriate terminal costs.
    \item We establish that cycles of the Riccati recursion correspond to periodic (non-stationary) Nash equilibria. Moreover, our simulations suggest that such cycles emerge, a phenomenon absent in the single agent case.
    \item We analyze the convergence of the Riccati recursion and show that it may converge to stationary equilibria, converge to periodic equilibria, or generate bounded but non-convergent trajectories, depending on system parameters, stability properties, and terminal costs.
\end{enumerate}

\section*{Notation}

 Let $\mathbb{R}^n$ denote the $n$-dimensional real vector space, and let $\mathbb{S}^n$ be the space of symmetric $n \times n$ matrices. The notation $\mathbb{S}^n_+$ (respectively, $\mathbb{S}^n_{++}$) denotes the cone of positive semidefinite (respectively, positive definite) symmetric matrices. We use bold symbols to indicate ordered tuples across $N$ agents: $\mathbf{P} = (P^1, P^2, \dots, P^N)$. We use $\mathbf{P}^{-i}$ to denote the ordered tuple $(P^1, \dots, P^{i-1}, P^{i+1}, \dots, P^N)$. We denote by $\| \cdot \|$ the matrix norm induced by the Euclidean norm on vectors, i.e., $\|A\| = \sup_{\|x\| = 1} \| A x \|$. This coincides with the spectral norm and is used throughout when referring to stability conditions.

\section{Finite and infinite-horizon LQ games} \label{sec:settings}

Consider the coupled linear dynamic system
\begin{equation}\label{eq:dynamics}
    x_{t+1} = A x_t + \sum_{i=1}^N B^i u_t^i + w_t,
\end{equation}
where $x_t \in \mathbb{R}^n$ is the global system state, $u_t^i \in \mathbb{R}^{m^i}$ is the control input of player $i \in [1, N]$, $w_t$ is Gaussian noise with zero mean and $\mathbb{E} [(w_t)^\top w_t] = W$. The matrices $A \in \mathbb{R}^{n \times n}$ and $B^i \in \mathbb{R}^{n \times m^i}$ describe the state and input dynamics, respectively. Each agent $i$ chooses a strategy $\gamma^{i}_t : \mathbb{R}^n \mapsto \mathbb{R}^{m^i} $. For the infinite-horizon case, we consider initially a stationary strategy $\gamma^i$.

At every time step, each agent $i$ aims to minimize its stage cost:
\begin{align*}
    c^i ( x, u^i ) = x^\top Q^i x + (u^i)^\top R^i u^i,
\end{align*}
where $Q^i, R^i \in \mathbb{S}^n_{++}$ \footnote{Unlike much of the literature, we assume $Q^i \succ 0$ (rather than $Q^i \succeq 0$). This choice simplifies detectability arguments; see Remark~\ref{rem:Q_pos_def}.}.

In the finite-horizon case, each player's cost is given by:
\begin{align} \label{eq:cost_finite_horizon}
    &J^i(x_0, \boldsymbol{\gamma}) = \mathbb{E}_{w} \left[ \frac{1}{T} \left( \sum_{t=0}^{T-1} c^i ( x_t,u^i_t ) + x_T^\top Q^i_T x_T   \right)\right]
\end{align}
with  $Q^i_T \in \mathbb{S}^{n}_{++}$. In the infinite-horizon case, the player's objective is
\begin{equation} \label{eq:cost_infinite_horizon}
    J^i(x_0,\boldsymbol{\gamma}) = \lim_{T \to \infty} \mathbb{E}_{w} \left[ \frac{1}{T}  \sum_{t=0}^{T} c^i ( x_t,u^i_t ) \right].
\end{equation}

\begin{definition} \label{def:NE}
    A set of strategies $\boldsymbol{\gamma}^*$ constitutes a Nash equilibrium if for all $i \in [1, N]$
\begin{align*}
        J^i(x_0, \gamma^{i,*},\boldsymbol{\gamma}^{-i,*}) \leq J^i( x_0, \gamma^{i},\boldsymbol{\gamma}^{-i,*}), \quad \forall \gamma^{i} \in \Gamma^{i},
\end{align*}
and, in the infinite-horizon case, the induced closed-loop dynamics is asymptotically stable.
\end{definition}
At a Nash equilibrium, no player can unilaterally reduce their cost by deviating from their equilibrium strategy. We focus on linear state-feedback strategies $u^i_t = K^i_t x_t$, with $K^i_t \in \mathbb{R}^{m^i \times n}$, for which the Nash equilibria can be computed by solving a set of coupled Riccati equations. While other (nonlinear) Nash equilibria might exist \cite{Basar98}, no general method is known for computing them in the LQ setting. In the next subsections, we analyze the finite-horizon case first, followed by the infinite-horizon one.

 Given a collection of feedback gains $\mathbf{K}_t$, we define the joint closed-loop matrix $A^{cl}_t$ and, for each agent $i$, the partial closed-loop matrix $A^{cl,-i}_t$:
\begin{align*}
    A^{cl}_t = A - \sum_{j= 1}^N B^j K^j_t, \quad A^{cl,-i}_t = A - \sum_{j \neq i} B^j K^j_t,
\end{align*}
Accordingly, we use $A^{cl}$ and $A^{cl,-i}$ to denote the infinite-horizon counterparts when the gains are time-invariant.
    

\section{Linear state feedback policies}

\subsection{Finite-horizon setting}

 The existence of a linear state-feedback Nash equilibrium in finite-horizon LQ games depends on the solvability of a set of coupled Riccati equations. Let $\mathbf{P}_t \in (\mathbb{S}_{+}^{n})^N$, for $t \in [0, T-1]$ and each $i \in [1,N]$ satisfy the equations
\begin{align}
    P_t^i = &Q^i +   (A^{cl,-i}_t)^\top P^i_{t+1} A^{cl,-i}_t - (A^{cl,-i}_t)^\top P^i_{t+1} B^i \nonumber \\
    &(R^i + (B^i)^\top P^i_{t+1} B^i)^{-1} (B^i)^\top P^i_{t+1} A^{cl,-i}_t, \label{eq:P_NE_FH} \\
    K_{t}^i = &(R^i + (B^i)^\top P^i_{t+1} B^i)^{-1} (B^i)^\top P^i_{t+1} A^{cl,-i}_{t}, \label{eq:K_NE_FH}
\end{align}
with $\mathbf{P}_T = \mathbf{Q}_T$. 

Note that the value of $K^i_t$ in \eqref{eq:K_NE_FH} depends on the policies $K^j_t$ of the other agents at the same time step through the term $A^{cl,-i}_t$. Equation \eqref{eq:K_NE_FH} can be rewritten as
\begin{align} \label{eq:riccati_eq_single}
    &R^i K^i_t + \sum_{j = 1}^N (B^i)^\top P^i_{t+1}  B^j K^j_t = (B^i)^\top P^i_{t+1} A.
\end{align}
By stacking together \eqref{eq:riccati_eq_single} for all the agents, we obtain a system of equations linear in $K_t$  shown in \eqref{eq:ricatti_eqn}. 

\begin{figure*}[ht]
{\small
\begin{equation}\label{eq:ricatti_eqn}
\begin{bmatrix}
R^1+(B^1)^\top P_{t+1}^{1} B^1 & (B^1)^\top P_{t+1}^{1} B^2 &  \cdots & (B^1)^\top P_{t+1}^{1} B^N \\
(B^2)^\top P_{t+1}^{2} B^1 & R^2+(B^2)^\top P_{t+1}^{2} B^2 & \cdots & (B^2)^\top P_{t+1}^{2} B^N \\
\vdots & \vdots & \ddots & \vdots \\
(B^N)^\top P_{t+1}^{N} B^1 & (B^N)^\top P_{t+1}^{N} B^2 & \cdots & R^N+(B^N)^\top P_{t+1}^{N} B^N
\end{bmatrix}
\begin{bmatrix}
    K_t^1 \\
    K_t^2 \\
    \vdots \\
    K_t^N
\end{bmatrix} = 
\begin{bmatrix}
    (B^1)^\top P_{t+1}^{1} A \\
    (B^2)^\top P_{t+1}^{2} A \\
    \vdots \\
    (B^N)^\top P_{t+1}^{N} A \\
\end{bmatrix}.
\end{equation}}
\end{figure*}

A finite-horizon linear quadratic game described in \eqref{eq:dynamics} and \eqref{eq:cost_finite_horizon} admits a unique feedback Nash equilibrium solution if and only if equation \eqref{eq:ricatti_eqn} has a unique solution for each $t \in [0, T-1]$ \cite[Corollary 6.1 and Remark 6.5]{Basar98}. In particular, if system \eqref{eq:P_NE_FH} \eqref{eq:K_NE_FH}  can be solved backward in time from $t = T-1$ to $t = 0$ without encountering singularities, then the resulting sequence $\{\mathbf{P}_t\}_{t=0}^T$ defines the only possible feedback Nash equilibrium of the game. Note that there are no known conditions on $(A, \mathbf{B}, \mathbf{Q}, \mathbf{R}, \mathbf{Q}_T)$ such that the solvability of \eqref{eq:ricatti_eqn} is guaranteed.

Since we are interested in studying the convergence properties of the Nash equilibrium, it is reasonable to assume that the game has a unique solution over every finite horizon $T$. Hence, we adopt the following assumption, consistent with \cite[Equation (15)]{engwerda1998open}, \cite[Assumption 2.3]{Hambly21}, and \cite[Assumption 1]{Roudneshin20}.
\begin{assumption} \label{ass:invertibility}
    Equation \eqref{eq:ricatti_eqn} admits a unique solution for any $t \in [0, \dots, T-1]$.
\end{assumption}

\subsection{Infinite-horizon setting}

We introduce the following assumption, necessary for the existence of a Nash equilibrium in the infinite-horizon game.
\begin{assumption} \label{ass:stab_detect}
    The pair $\left(A, \left[ B^1 \dots B^N \right] \right)$ is stabilizable.
\end{assumption}
If the pair $\left(A, \left[ B^1 \dots B^N \right] \right)$ is not stabilizable, there does not exist a set of strategies for which the induced closer-loop dynamics is asymptotically stable.

Let $\mathbf{P} \in (\mathbb{S}_{+}^{n})^N$ satisfies the following set of equations for $i \in [1,N]$:
\begin{align}
      P^i = &Q^i +   (A^{cl,-i})^\top P^i A^{cl,-i} - (A^{cl,-i})^\top P^i B^i  \nonumber\\
      &(R^i + (B^i)^\top P^i B^i)^{-1} (B^i)^\top P^i A^{cl,-i} \label{eq:P_NE_IH},\\
       K^i = &(R^i + (B^i)^\top P^i B^i)^{-1} (B^i)^\top P^i A^{cl,-i}. \label{eq:K_NE_IH}
    \end{align}
As in the finite-horizon case, equation \eqref{eq:K_NE_IH} can be rewritten as:
\begin{align} \label{eq:riccati_eq_single_IH}
    & R^i K^i + \sum_{j = 1}^N (B^i)^\top P^i  B^j K^j = (B^i)^\top P^i A
\end{align}
and stacked with the ones of the other agents to obtain an equation analogous to \eqref{eq:ricatti_eqn}. Thus, we introduce the function $g: (\mathbb{S}_{+}^{n})^N \mapsto \{\mathbb{R}^{m_i \times n} \}_{i=1}^N$ defined implicitly by \eqref{eq:K_NE_IH}, and rewrite $\mathbf{K} = g(\mathbf{P})$. In this way $\mathbf{K}$ can be substituted into \eqref{eq:P_NE_IH}, yielding a system solely in terms of $\mathbf{P}$. For each given set of solutions $\mathbf{P}$, the feedback policies $\mathbf{K}$ are uniquely determined.

From \cite[Proposition 6.3]{Basar98} we know that, for any solution $\mathbf{P} \in (\mathbb{S}_{+}^{n})^N$ of the set of equations \eqref{eq:P_NE_IH} and \eqref{eq:K_NE_IH} such that each pair $(A^{cl,-i}, B^i)$ is stabilizable and $(A^{cl,-i}, Q^i)$ is detectable, the policies $\mathbf{K}  = g\left( \mathbf{P} \right)$ stabilize the closed-loop system and form a Nash equilibrium.
Notice that in our setting, the pair $(A^{cl,-i}, (Q^i)^{1/2})$ is always detectable for any $i$, since we assumed $Q^i$ positive definite. Using the one-to-one correspondence between $\mathbf{P}$ and $\mathbf{K}$, we formally define the set of stationary linear state-feedback Nash equilibria in terms of $\mathbf{P}$:
\begin{align*}
    \mathcal{P}^{NE}_{stat} := \left\{ \mathbf{P}: \mathbf{K} = g\left( \mathbf{P} \right) \text{ is a Nash equilibrium} \right\}.
\end{align*}
Sufficient conditions on the existence of a Nash equilibrium have been obtained in the scalar case in \cite{salizzoni2024nash, nortmann2023feedback}, but have not been generalized to multidimensional state space settings.

\section{Connections between finite and infinite-horizon Nash equilibria} \label{sec:connections}

In the single-agent case (the LQR), the relationship between finite and infinite-horizon solution is well understood: the backward Riccati recursion always converges to the unique solution of the algebraic Riccati equation, which characterizes the infinite-horizon optimal policy. Consequently, the terminal cost only influences the transient part of the finite-horizon solution, while the asymptotic behavior is uniquely determined.

In LQ games, the situation is fundamentally different for three reasons. First, the infinite-horizon problem may admit multiple equilibria. Second, these equilibria need not be stationary: periodic orbits are a common feature of discrete maps \cite[Chapter 7]{katok1995introduction}, so the recursion may generate cycles, corresponding to periodic Nash equilibria. Third, the recursion may not converge. In this Section we study the first two points, while in the next one we study the third through simulations.

Consider the set \(\mathbf{P}_t\), evolving backward in time according to the equations \eqref{eq:P_NE_FH} and \eqref{eq:K_NE_FH} which can be written concisely using the function $f : ( \mathbb{S}^{n}_{++} )^N \to ( \mathbb{S}^{n}_{++} )^N $:
\begin{align} \label{eq:iteration}
\mathbf{P}_t  = f\left( \mathbf{P}_{t+1} \right), \quad \mathbf{P}_T = \mathbf{Q}_T.
\end{align}
Under Assumption \ref{ass:invertibility}, function \(f\) returns a unique solution. Moreover, notice that $f$ preserves positive definiteness: if  $\mathbf{P}_{t+1} \in ( \mathbb{S}^{n}_{++} )^N$ , then also $\mathbf{P}_{t} \in ( \mathbb{S}^{n}_{++} )^N$, from equations \eqref{eq:P_NE_FH} and \eqref{eq:K_NE_FH} we obtain:
\begin{align} \label{eq:alterPFH}
    P_t^i = &Q^i +   (K^i_t)^\top R^i K^i_t + (A^{cl}_t)^\top P^i_{t+1} A^{cl}_t,
\end{align}
where all the terms on the right-hand side are positive semidefinite. 



\subsection{Fixed points}

The fixed points of the function $f$ are:
\begin{align} \label{def:fixed_points}
    \text{Fix}(f) := \left\{ \mathbf{P} \in (\mathbb{S}^{n}_{++} )^N: \mathbf{P} = f\left( \mathbf{P} \right) \right\}
\end{align}
The following result establishes the equivalence between these fixed points and the stationary Nash equilibria of the infinite-horizon game.
\begin{proposition}
\label{pr:equilibrium_equivalence}
Consider the set $\mathcal{P}^{NE}_{stat}$ of the infinite-horizon LQ game described by \eqref{eq:dynamics} and \eqref{eq:cost_infinite_horizon}, and the iteration \eqref{eq:iteration}, with the same matrices $\mathbf{Q}$ and $\mathbf{R}$ as the infinite-horizon LQ game. Then
\begin{equation}
\label{eq:fixed_point}
\mathcal{P}^{NE}_{stat} = \text{Fix}(f).
\end{equation}
\end{proposition}
\begin{proof}
The proof that $ \mathcal{P}^{NE}_{stat} \subseteq \text{Fix}(f)$ follows from the definitions and the structure of the equations. Specifically, the infinite-horizon Nash equilibrium conditions \eqref{eq:P_NE_IH} and \eqref{eq:K_NE_IH} correspond to the fixed point limit of the recursion defined by \eqref{eq:P_NE_FH} and \eqref{eq:K_NE_FH}  as $t \to \infty$ and are equivalent to the equation $\mathbf{P} = f \left( \mathbf{P} \right)$. Thus, any \(\mathbf{P}\) that solves the infinite-horizon equations also satisfies the fixed-point condition \eqref{def:fixed_points} of the recursion \eqref{eq:iteration}.\\
To prove that $ \text{Fix}(f) \subseteq \mathcal{P}^{NE}_{stat}$ we need to prove additionally that each pair $[A^{cl,-i}, B^i]$ is stabilizable, and each pair $[A^{cl,-i}, (Q^i)^{1/2}]$ is detectable . By combining equations \eqref{eq:P_NE_IH} and \eqref{eq:K_NE_IH} we obtain:
\begin{align}
    (A^{cl})^\top P^i A^{cl} - P^i = -(Q^i + (K^i)^\top R^i K^i) \prec 0,
\end{align}
since $Q^i, R^i \succ 0$. Hence $\rho(A^{cl,-i} + B^i K^i) =\rho(A^{cl}) < 1$, which proves that the pair $[A^{cl,-i}, B^i]$ is stabilizable. Given that each $Q^i$ is positive definite, then each pair  $[A^{cl,-i}, (Q^i)^{1/2}]$ is detectable.
\end{proof}
\begin{remark}\label{rem:Q_pos_def}
   If only $Q^i\succeq 0$, the inclusion $\mathrm{Fix}(f)\subseteq \mathcal{P}^{\mathrm{NE}}_{\mathrm{stat}}$ may fail. The difficulty is that the closed-loop matrices $A^{\mathrm{cl},-i}$ depend on the fixed point $\mathbf{P}=f(\mathbf{P})$, and we do not currently have conditions ensuring detectability of $(A^{\mathrm{cl},-i},(Q^i)^{1/2})$. Without detectability, the Riccati equation \eqref{eq:P_NE_IH} for player $i$ may admit multiple positive semidefinite solutions, and the optimal one may not be stabilizing \cite{kuvcera1972discrete}. Consequently, the fixed-point condition alone does not ensure a Nash equilibrium. In applications one may enforce this by a small regularization $Q^i \leftarrow Q^i + \epsilon I$ (which perturbs the game), or verify the detectability at the computed fixed point.
\end{remark}
A consequence of the above proposition is the following corollary, which explicitly links each infinite-horizon Nash equilibrium to the one of a finite-horizon game through appropriate selection of terminal costs.
\begin{corollary}
\label{cor:terminalCost}
Consider an infinite-horizon LQ game described by \eqref{eq:dynamics} and \eqref{eq:cost_infinite_horizon}. For any \(\mathbf{P} \in \mathcal{P}^{NE}_{stat}\), we can design a finite-horizon LQ game described by \eqref{eq:dynamics} and \eqref{eq:cost_finite_horizon} whose finite-horizon Nash equilibrium \( \mathbf{K}_t \) exists, is stationary,  stabilizes the system, and satisfies
\begin{equation}
\label{eq:terminal_cost_choice}
\mathbf{K}_t = \mathbf{K} = g\left(\mathbf{P} \right)\quad \forall t \in [0,1,\dots,T-1],
\end{equation}
by setting the the terminal costs as \(\mathbf{Q}_T = \mathbf{P}\).
\end{corollary}
\begin{proof}
Consider selecting the terminal costs as \(\mathbf{Q}_T = \mathbf{P}\). Given this choice, the finite-horizon Riccati recursion \eqref{eq:iteration} immediately satisfies the equilibrium condition \eqref{eq:fixed_point} at the terminal step. Consequently, this choice results in a constant solution \(\mathbf{P}_t = \mathbf{P}\) for all \(t \in [0,T]\), implying that the finite-horizon Nash equilibrium \eqref{eq:K_NE_FH} matches the infinite-horizon equilibrium exactly at every time step. Moreover, this solution stabilizes the system since at every time step $\rho \left( A - \sum_{i=1}^N B^i K^i_t \right) = \rho \left( A - \sum_{i=1}^N B^i K^i \right) < 1 $.
\end{proof}
By tuning the terminal costs, it is possible to obtain all the Nash equilibria of the infinite-horizon games in the finite-horizon version. However, achieving this requires prior knowledge of the exact infinite-horizon equilibria, which are often difficult to compute.

\subsection{Cycles} \label{subsec:cycles}

\begin{figure}[ht]
    \hspace*{-0.1cm}
    \includegraphics[width=\linewidth]{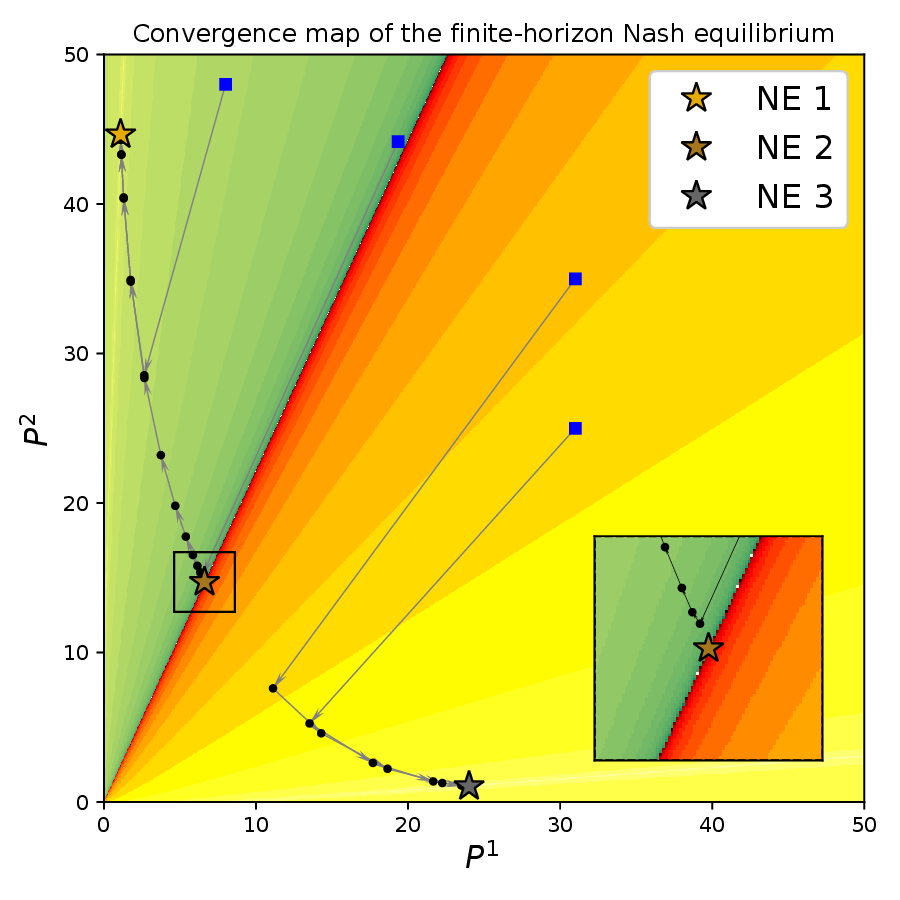}
    \caption{Convergence of the iteration based on the terminal costs. We did a grid search with $10^6$ combinations of  $\mathbf{Q}_T$. The three stars represent the $\mathbf{P}$ matrices of the three  Nash equilibria. If the terminal costs are in the green area, the iteration converged to the first Nash equilibrium, if they are in the yellow-red area, it converged to the third  Nash equilibrium. The darker the color, the more steps it took to converge. The image shows also the evolution of the recursion for four different initial conditions, highlighted by the blue squares.}
    \label{fig:convergence}
\end{figure}

\begin{figure*}[ht!]
  \centering
  \includegraphics[width=\textwidth]{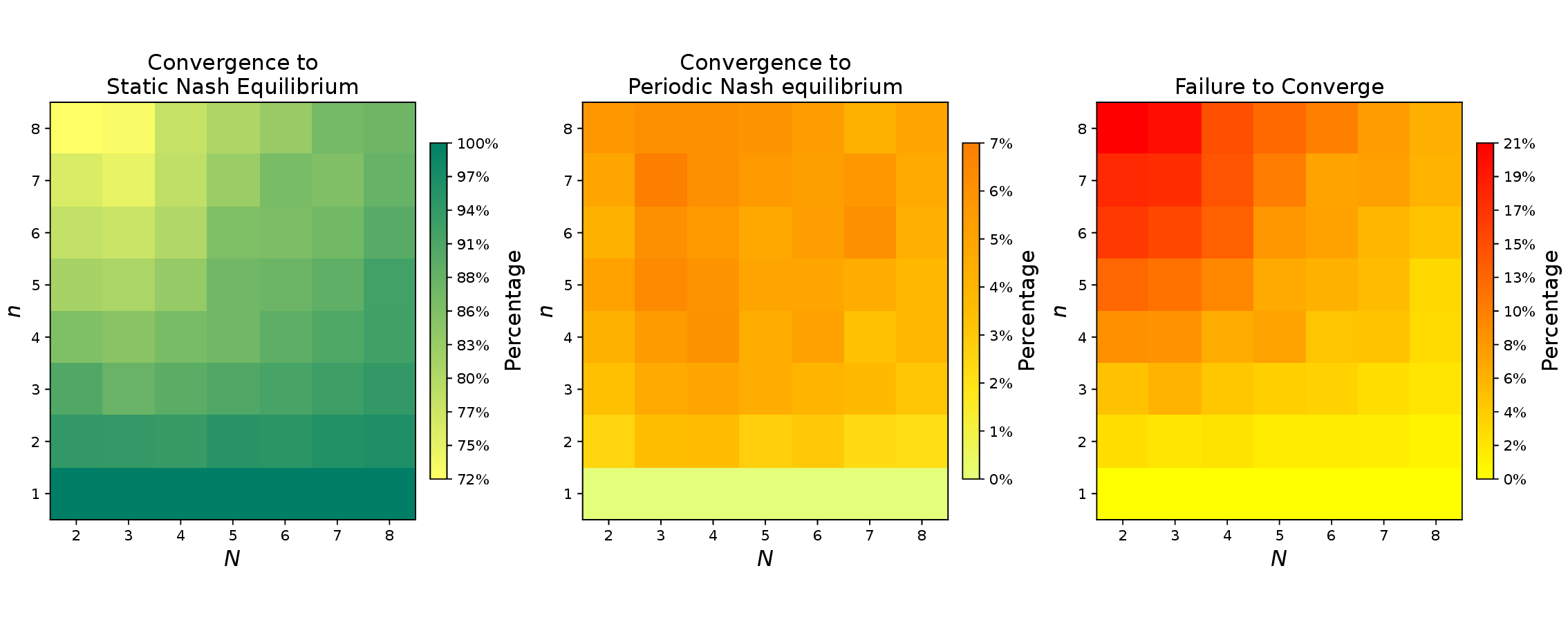} 
  \caption{For each pair $(n,N)$ we randomly generated $10^5$ games with random terminal costs. The value of $d$ is equal to $n$, but the results do not change significantly for other values. The three tables show the percentage of cases converging to an equilibrium point (left), to a cycle (center), or not converging (right). Notice that three different scales are used to better highlight variations.}
  \label{fig:table}
\end{figure*}

The nonlinear nature of recursion \eqref{eq:iteration} can admit other asymptotic behaviors. In particular, a set of symmetric positive semidefinite matrices $\{ \mathbf{P}_l \}_{l=1}^L$, with $\mathbf{P}_l \in (\mathbb{S}^{n}_{++} )^N$, constitutes a cycle if 
\begin{align} \label{eq:condition_cycle}
    &\mathbf{P}_l = f(\mathbf{P}_{l+1}), \quad \text{for } l \in [1,L-1], \quad \mathbf{P}_L = f(\mathbf{P}_{1}).
\end{align}
A cycle can have any length $L \in [2, \infty)$. We first prove that such cycles stabilize the system.
\begin{lemma} \label{lemma:cycle_stability}
  Given a cycle $\{ \mathbf{P}_l \}_{l=1}^L$ defined by \eqref{eq:condition_cycle}, the sequence of matrices $ \mathbf{K}_l = g \left( \mathbf{P}_l \right)$ for $l \in [1,L]$ stabilize the closed-loop system:
\begin{align*}
  \rho  \left( \prod_{l = 1}^L A^{cl}_l\right) < 1.
\end{align*}
\end{lemma}
\begin{proof}
By combining Equations \eqref{eq:P_NE_FH} and \eqref{eq:K_NE_FH} we obtain, for $i \in [i,N]$:
\begin{align*}
    P^i_l = Q^i + (K^i_l)^\top R^i K^i_l + (A^{cl}_l)^\top P^i_{l+1}A^{cl}_l.
\end{align*}
Defin $\Theta_0 := I_n$ and $\Theta_{l} := A^{cl}_l A^{cl}_{l-1} \dots A^{cl}_1 $ for $l \geq 1$.  Iterating the previous identity over one period gives:
\begin{align} \label{eq:ploop}
        P^i_{1} = \sum_{l = 1}^{L} \Theta_{l-1}^\top \left( Q^i + (K^i_l)^\top R^i K^i_l \right) \Theta_{l-1} + \Theta_L^\top P^i_1 \Theta_L.
\end{align}
Suppose, by contradiction, that $\Theta_L$ has an eigenvector $v \in \mathbb{R}^n, \|v\|=1$ with an associated eigenvalue $\lambda \in \mathbb{R}, |\lambda| > 1$. Let $v_\ell:=\Theta_\ell v$. Left- and right-multiplying \eqref{eq:ploop} by $v^\top$ and $v$ gives
\begin{align} \label{eq:lambdaeq}
    (\lambda^2 - 1) v^\top P^i_1 v = - \sum_{l= 1}^L v_l^\top \left( Q^i + (K^i_l)^\top R^i K^i_l \right)  v_l .
\end{align}
Since by assumption $\lambda^2 - 1 \geq 0$ and $P^i_1 \in \mathbb{S}^{n}_{++} $, the term on the left is positive. However, since both $Q^i$ and $R^i$ are positive definite, the right-hand side is strictly negative, which is a contradiction. Therefore $|\lambda|<1$ for every eigenvalue of $\Theta_L$.
\end{proof}
We can now prove that a cycle is a Nash equilibrium of the infinite-horizon game.
\begin{theorem}
\label{th:loopNE}
Consider a cycle $\{ \mathbf{P}_l \}_{l=1}^L$, with $\mathbf{P}_l \in (\mathbb{S}^{n}_{++} )^N$, which satisfies condition \eqref{eq:condition_cycle}, and let Assumption \ref{ass:stab_detect} hold. Then, $\{ \mathbf{K}_l \}_{l=1}^L$, with $ \mathbf{K}_l = g \left( \mathbf{P}_l \right)$, constitutes a Nash equilibrium of the infinite-horizon LQ game described in \eqref{eq:dynamics} and \eqref{eq:cost_infinite_horizon}.
\end{theorem}
\begin{proof}
An LQR with a state matrix that is periodically time variant has a unique solution under the conditions provided in \cite[Theorem 6.11]{bittanti1991periodic}. In the case of cycles in LQ games, each agent $i$ is facing a LQR problem with a periodic state matrix, $A^{cl,-i}_{Lt + l} = A - \sum_{j \neq i}B^j K^j_l$ for any $t \in \mathbb{N}_+$ and any $l \in [0,L-1]$. Since the periodic matrix $A^{cl}_l = A - \sum_{i = 1}^N B^i K^i_l$ is stable, as proved in Lemma \ref{lemma:cycle_stability}, the pair $[A^{cl,-i}_l, B^i]$ is stabilizable for every agent $i$. Moreover, given that each $Q^i$ is positive definite, each pair $(A^{cl,-i}, (Q^i)^{1/2})$ is detectable. Then, for every agent $i$, the set of matrices $\{P^i_l \}_{l=1}^L$ solves the discrete periodic Riccati equation and is thus the unique symmetric periodic positive semidefinite solution \cite[Theorem 6.11 (iv)]{bittanti1991periodic}.
\end{proof}
To the best of our knowledge, such non-stationary equilibria have not previously been identified in time-invariant LQ games. While we were not able to provide an analytical result on the existences of periodic Nash equilibria, in the next section we will present example of games whose finite-horizon Nash equilibrium converges to a cycle, providing empirical evidence that periodic Nash equilibria may exist.

\section{Simulations}

In the LQR case, the recursion \eqref{eq:iteration} always converges to the unique solution of the algebraic Riccati equation, regardless of the terminal cost. This guarantees that finite-horizon formulations consistently approximate the infinite-horizon problem.

In games, the situation is more delicate. As shown in Section \ref{sec:connections}, infinite-horizon Nash equilibria may be multiple and can be either stationary or periodic. Consequently, the finite-horizon recursion may converge to different equilibria depending on the terminal costs, but it may also fail to converge altogether. Understanding this behavior is critical for assessing the reliability of finite-horizon approximations of infinite-horizon games.

Analytical characterization of global convergence remains elusive. Even in the continuous-time scalar two-agent case, \cite{weeren99} showed that convergence analysis connects with open problems in nonlinear dynamics. As a result, an analytical characterization remains out of reach. Instead, we rely on numerical experiments to complement theoretical insights.

\subsection{Convergence to the fixed points} \label{subsec:fixedpoints}

The fixed points of recursion \eqref{eq:iteration} coincide with the stationary Nash equilibria of the infinite-horizon game (Proposition \ref{pr:equilibrium_equivalence}). Their stability determines whether they can be attained by the recursion: attractive fixed points draw nearby trajectories, while unstable ones repel them.

In the continuous-time scalar two-agent case, \cite{weeren99} proved that at least one attractive fixed point always exists, and that the recursion converges for any terminal cost. Inspired by this analysis, we conducted numerical experiments to investigate whether analogous stability properties hold in discrete-time LQ games. We analyze how different choices of terminal cost matrices \( \mathbf{Q}_T \) affect the evolution of \eqref{eq:iteration}.
Consider the scalar two-agent infinite-horizon game with parameters chosen for simplicity and clarity: $A = 5$, $B_1 = B_2 = Q_1 = Q_2 = R_1 = 1$, and $R_2 = 2$. Using the algorithm proposed in~\cite{salizzoni2024nash}, we compute the three Nash equilibria of this game, together with the corresponding matrices $\mathbf{P}$. In this example, the recursion consistently converges for any initial condition $\mathbf{Q}_T$ to one of the Nash equilibria. Two of them are attractive fixed points, while one behaves as a saddle point. Fig \ref{fig:convergence} summarizes the results obtained.

We also ran multiple simulations to study how often the recursion \eqref{eq:iteration} would converge to a fixed point for different values of $n$, $d$, and $N$. The results are presented in the leftmost table in Fig \ref{fig:table}. We simulated for each cell $10^4$ games by randomly sampling the matrices $A, B^i, Q^i$ and $R^i$. In the scalar case, $n = 1$, the recursion converges all the time. The percentage of convergence decreases as the dimension of the state $n$ increases. 
\begin{figure}[ht]
    \hspace*{-0.1cm}
    \includegraphics[width=\linewidth]{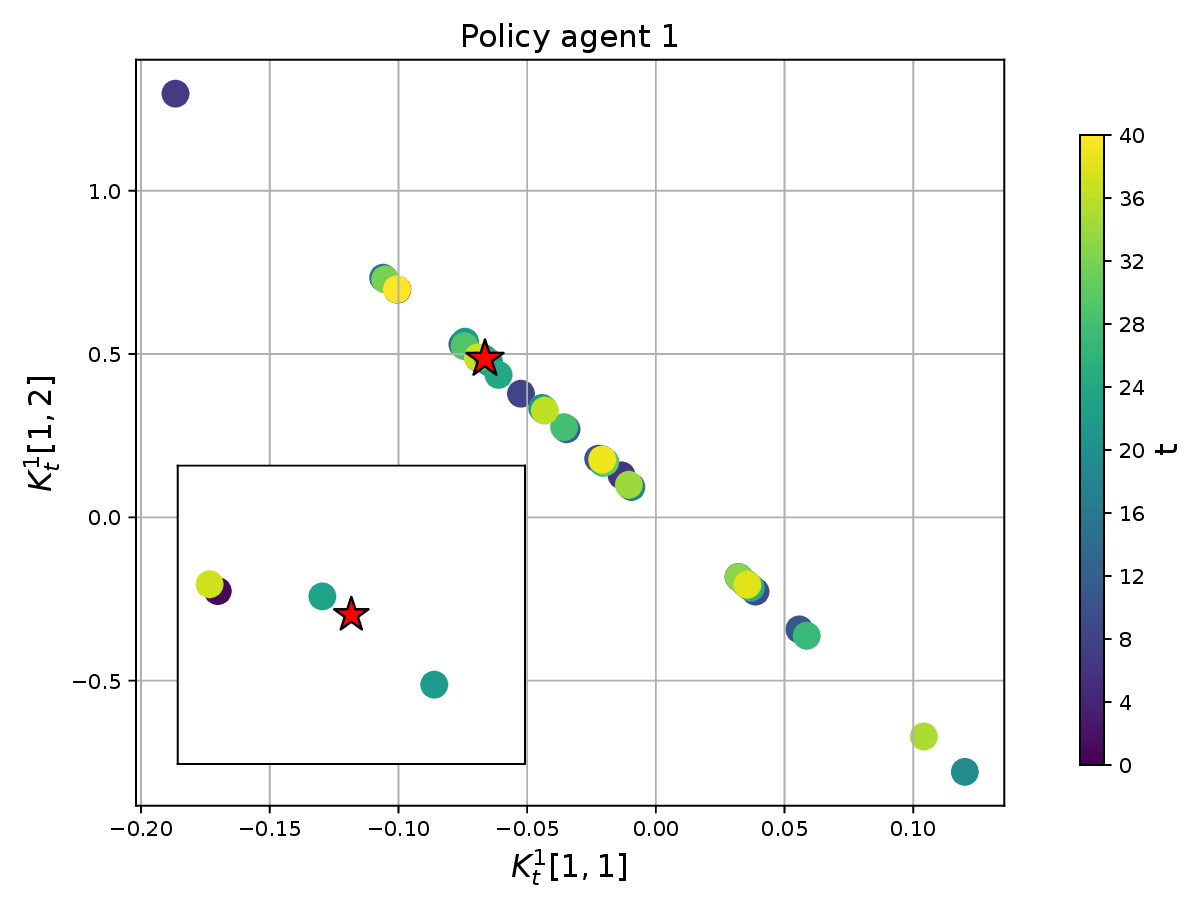}
    \caption{Evolution of the policy of agent $1$ for a randomly generated game with $n = 2$, $d = 1$, and $N = 3$. Although the game horizon was $T = 1000$, we plotted only the $K^1_t$ for $t$ in $[0,40]$ to improve visualization clarity. The game admits a Nash equilibrium (red star), computed using gradient descent, which, however, is unstable for the recursion \eqref{eq:iteration}. }
    \label{fig:NonConverging}
\end{figure}
The recursion \eqref{eq:iteration} did not converge in a non-negligible fraction of instances (as reported in the right panel of Fig. \ref{fig:table}). When convergence fails, trajectories remain bounded but do not settle to an equilibrium. An example is shown in Fig. \ref{fig:NonConverging}, where the policies oscillate without converging, despite the existence of a Nash equilibrium computed by gradient descent. This discrepancy can arise either because all equilibria of the game are unstable for the recursion, or because attractive equilibria exist but the selected terminal costs lie outside their basins of attraction. This demonstrates that convergence cannot be inferred from existence alone, but requires stability of the equilibria.

\subsection{Convergence to cycles} \label{subsec:conv_cycles}

Beyond fixed points, recursion \eqref{eq:iteration} may converge to periodic Nash equilibria as established in Section \ref{sec:connections}. These equilibria are non-stationary but stabilizing, and represent a novel phenomenon in time-invariant LQ games. As the fixed points, also the cycles may be attractive or unstable.

Our simulations show that convergence to cycles is relatively common (see Fig \ref{fig:table}, center panel). In Fig \ref{fig:lenght} we show how frequently attractive cycles of a certain lengths appear for various combinations of $n$, $d$, and $N$. As expected, longer cycles occur less frequently. Interestingly, the cycle length appears to be largely independent of the system’s dimensionality or the number of agents.
\begin{remark}
    This numerical behavior alone does not constitute a proof of existence of periodic Nash equilibria, since long transients or slow drift can produce near-periodicity for very long horizons. Proving existence in a given instance would require solving the algebraic conditions \eqref{eq:condition_cycle}.
\end{remark}

\begin{figure}[ht]
    \hspace*{-0.1cm}
    \includegraphics[width=\linewidth]{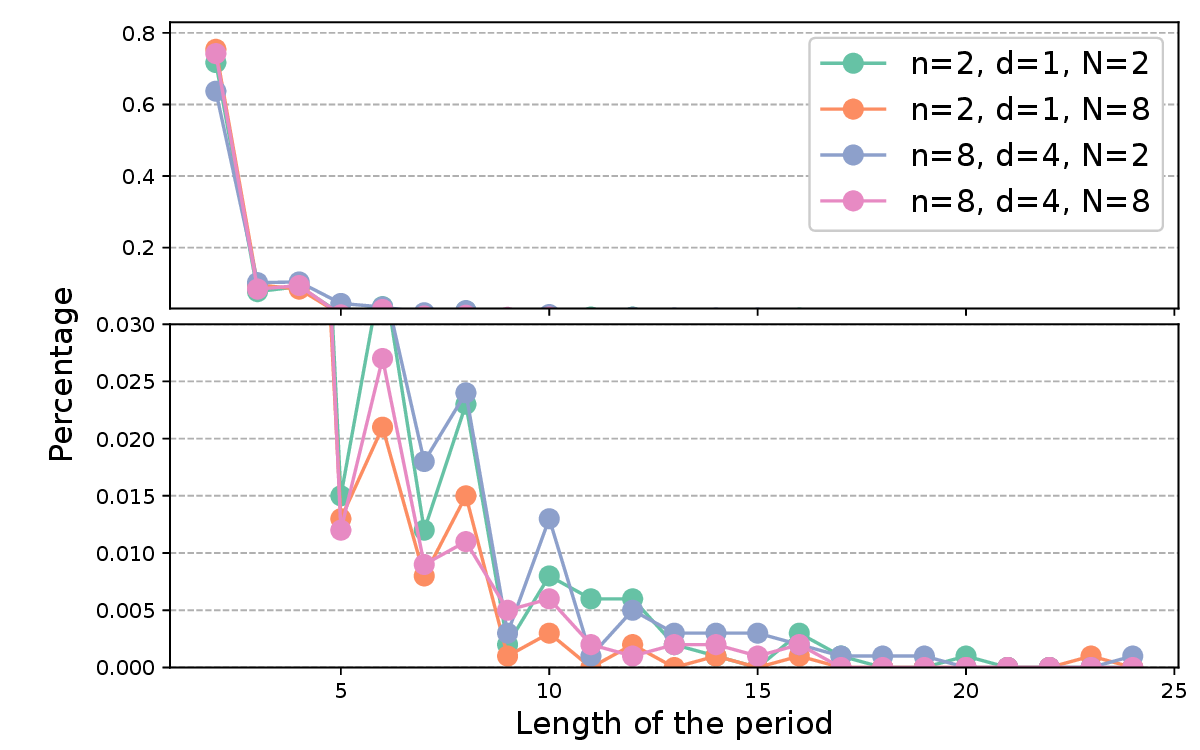}
    \caption{For each case, we kept randomly generating games until we found $10^3$ games whose Nash equilibrium converged to a cycle. In the plot we show the distribution of the length of the cycles from this set.}
    \label{fig:lenght}
\end{figure}
An illustrative example is provided in Fig. \ref{fig:loop} for a randomly generated game, for which the recursion converges to a cycle of length $L = 5$. Notice that the spectral radius of the closed-loop system is not always stable. In particular, there is a set $\mathbf{K}_l$ for which $\rho(A^{cl}_l) = 3.2$. In fact, while Lemma \ref{lemma:cycle_stability} establishes that the cycle as a whole is stabilizing, policies within the cycle are not necessarily stabilizing. In particular, it may occur that for some $l \in [1,L]$ the corresponding closed-loop matrix satisfies $\rho(A^{cl}_l) \geq 1$. This observation is relevant for receding-horizon games, where only one policy at a time is applied, and hence transient instabilities may arise even though the overall cycle guaranties asymptotic stability.
\begin{figure}[ht]
    \hspace*{-0.1cm}
    \includegraphics[width=\linewidth]{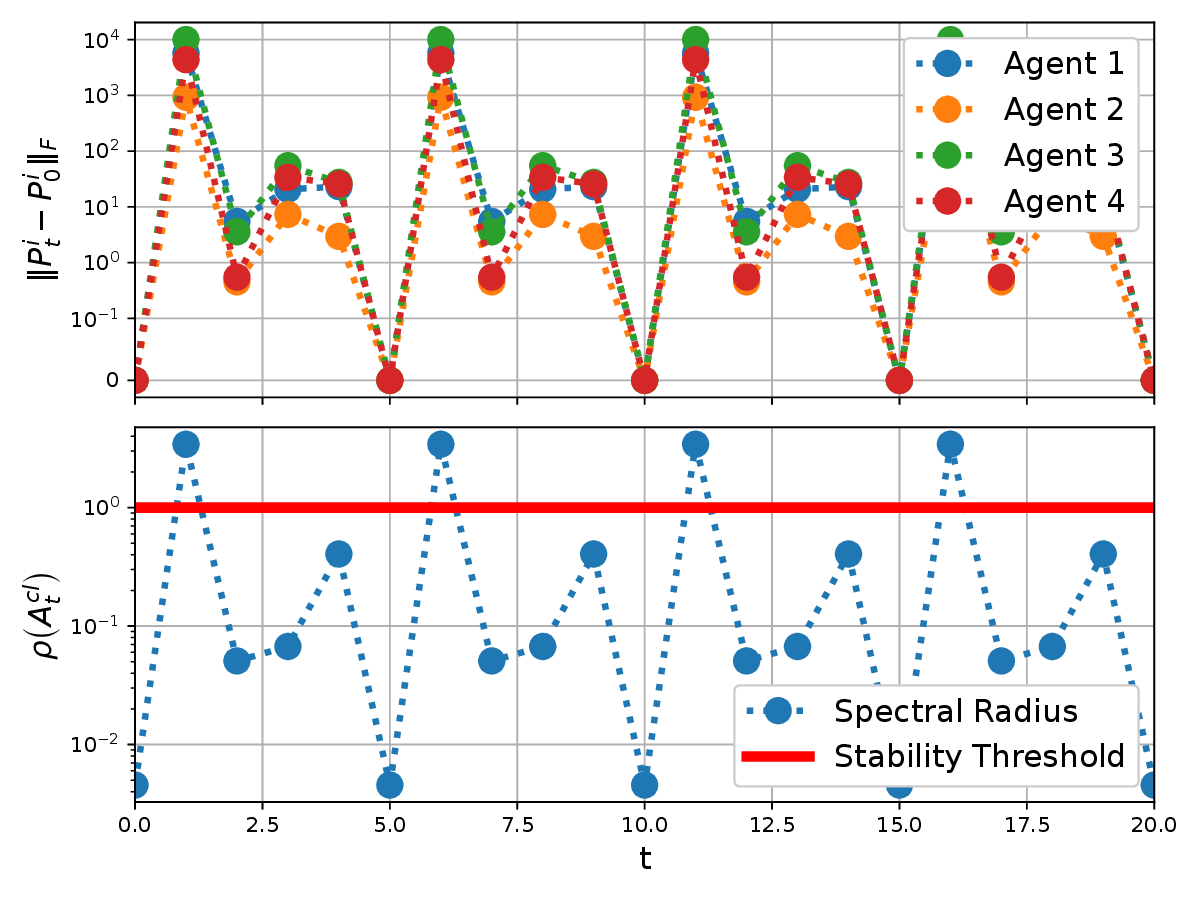}
    \caption{Example of a periodic Nash equilibrium of length $5$  for a randomly generated game with $n = 3$, $d = 2$, and $N = 4$. In the top panel we plotted the difference in the Frobenius norm $\|P^i_{t+1}-P^i_0\|_F$. In the bottom panel, the spectral radius of the closed-loop system $A^{cl}_t$.
    }
    \label{fig:loop}
\end{figure}

\section{Conclusions}

In this paper, we established the connection between finite-horizon and infinite-horizon Nash equilibria in LQ games, by studying the finite-horizon Nash equilibrium as a dynamical system. First, we showed that fixed points of the system are the infinite-horizon Nash equilibria and that any such equilibrium can be reached by choosing the finite-horizon terminal costs accordingly. Second, we proved that cycles of the recursion, if present, are periodic Nash equilibria for the time-invariant game. Third, through extensive simulations we documented three qualitatively distinct outcomes of the recursion: convergence to stationary equilibria, convergence to periodic equilibria, and bounded non-convergent trajectories.

These findings have practical consequences for receding-horizon games. Given that the terminal costs act as equilibrium selectors and attractive sets may be small, performance and stability depend on how receding-horizon problems are posed. Moreover, while a periodic equilibrium is stabilizing over one period, individual policies within the cycle need not be stabilizing at each step, which may cause instability when strategies are applied one step at a time.


\bibliographystyle{IEEEtran}
\bibliography{ref}

\end{document}